\documentclass[sigconf]{acmart}
\usepackage{microtype}
\usepackage{graphicx}
\usepackage{subcaption}
\usepackage{booktabs} 
\usepackage{hyperref}
\usepackage{booktabs} 
\usepackage{multirow} 
\usepackage{graphicx}
\usepackage{amsmath}

\usepackage{amssymb}
\usepackage{mathtools}
\usepackage{amsthm}
\usepackage{algorithm}
\usepackage{algpseudocode}
\usepackage[capitalize,noabbrev]{cleveref}
\AtBeginDocument{%
  }

\begin{document}

\title{Just Talk Once: Communication-Efficient Split Federated LLM Fine-Tuning on Edge Devices}

\author{Jiaxiang Geng\textsuperscript{1,2}, Xianhao Chen\textsuperscript{2}, Bing Luo\textsuperscript{1}}

\affiliation{
  \institution{\textsuperscript{1}Duke Kunshan University, Jiangsu, China
  \quad
  \textsuperscript{2}The University of Hong Kong, Hong Kong, China}
  \country{}
}

\email{{jg645,bl291}@duke.edu, xcheneee@hku.hk}

\thanks{The work of X. Chen was supported in part by the Research Grants Council of Hong Kong under Grant 27213824, Grant 17207826, and Grant CRS HKU702/24. The work of B. Luo was supported by Suzhou Frontier Science and Technology Program (Project SYG202310). \\ Corresponding Author: Bing Luo.}

\renewcommand{\shortauthors}{Jiaxiang Geng, Xianhao Chen, Bing Luo}

%\begin{abstract}
%The imminent exhaustion of public datasets has shifted the focus of Large Language Model (LLM) fine-tuning toward leveraging massive amounts of private data generated on edge devices. While Split Federated Learning (SFL) enables collaborative training, existing privacy-preserving architectures (e.g., U-shaped SFT) suffer from a prohibitive bidirectional communication loop, rendering them inefficient for unstable edge environments. To resolve this, we propose L-shaped SFT, a novel unidirectional framework that achieves extreme communication efficiency while preserving strict privacy. By leveraging the semantic consistency of LLMs, we offload loss computation to the server without exposing private labels and introduce activation-based loss. To protect intermediate activations against inversion attacks, we introduce a lightweight dual-matrix obfuscation mechanism that rotates the semantic space using isometric transformations. Furthermore, we propose a one-shot workflow that allows clients to upload only once and remain idle for the afterwards training duration. Extensive experiments on MMLU and WikiText-2 demonstrate that our approach reduces communication overhead by $56\%$ compared to U-shaped baselines, establishing a new state-of-the-art for private, bandwidth-efficient edge LLM fine-tuning.
%\end{abstract}

\begin{abstract}
Large language model (LLM) fine-tuning is increasingly shifting toward data generated on edge devices, where memory, computation, bandwidth, and connectivity constraints make conventional federated learning difficult to sustain. Split federated fine-tuning (SFT) improves client-side efficiency by offloading most model parameters and computation to the server but requires step-by-step bidirectional communication loop across the split interface and forces continuous client involvement throughout training. In this paper, we present \textbf{L-shaped SFT}, a split fine-tuning framework that removes this bidirectional bottleneck. Our key insight is that weight tying in modern LLMs enables server-side hidden activations to be directly supervised using target embeddings, allowing the training loss to be computed on the server without returning server outputs to the client. To further eliminate the need for continuous client participation, based on L-shaped SFT, we introduce \textbf{one-shot SFT}, in which clients upload activations once and then go offline while the server continues optimization over cached representations. We implement our design in a real system testbed with heterogeneous edge clients, including commercial smartphones and NVIDIA developer boards. Experiments demonstrate that our schemes significantly reduce communication costs and client online time compared with existing SFT baselines.
\end{abstract}

\begin{CCSXML}
<ccs2012>
   <concept>
       <concept_id>10010147.10010178.10010219</concept_id>
       <concept_desc>Computing methodologies~Distributed artificial intelligence</concept_desc>
       <concept_significance>500</concept_significance>
       </concept>
 </ccs2012>
\end{CCSXML}

\ccsdesc[500]{Computing methodologies~Distributed artificial intelligence}

\keywords{large language model, split federated fine-tuning, edge system}

\maketitle

\section{Introduction}

In recent years, Large language models (LLMs) have achieved remarkable success across various domains \cite{openai2024gpt4technicalreport, touvron2023llama, gemmateam2025gemma3technicalreport}. The scaling of these models faces a critical bottleneck: public datasets are projected to be exhausted between 2026 and 2032 \cite{DBLP:journals/corr/abs-2001-08361, 10.5555/3692070.3694094}. Consequently, the focus of fine-tuning is shifting toward leveraging massive amounts of data generated on edge devices \cite{gunter2024appleintelligencefoundationlanguage}. However, directly collecting such data for centralized fine-tuning is often undesirable due to ownership, deployment, and compliance constraints, especially in sensitive application domains such as finance \cite{wu2023bloomberggptlargelanguagemodel} and healthcare \cite{Thirunavukarasu2023}. This makes edge-native fine-tuning an important systems direction for practical LLM customization.

Federated learning (FL) provides a natural way to leverage decentralized data without transferring raw samples \cite{9579038}. However, directly applying conventional FL to LLM fine-tuning is highly challenging on resource-constrained edge devices \cite{10944288,bai2024federated}. In conventional FL, each client must maintain and train a full model locally, and then transmit full-model updates to the server for aggregation, as illustrated in Fig.~\ref{intro}(a). For LLM fine-tuning, this workflow creates three major bottlenecks in edge systems: \textit{memory}, \textit{computation}, and \textit{communication}. From the \textit{memory perspective}, local full-model fine-tuning is often infeasible. FP16 training requires roughly 16 GB of RAM per 1B parameters, and the requirement doubles for FP32 training \cite{McKeag2025}. In contrast, practical edge devices offer far less memory---for example, NVIDIA Jetson Orin Nano developer boards provide only 4 GB RAM, and commercial smartphones such as the Huawei nova 9 Pro provide only 8 GB RAM. Running and fine-tuning a full LLM on such devices can therefore easily lead to out-of-memory failures. From the \textit{computation perspective}, edge devices typically lack the processing capability needed for efficient full-model training. Even when memory is sufficient, updating all parameters locally can take substantial time, making conventional FL difficult to sustain in latency-sensitive edge environments. From the \textit{communication perspective}, conventional FL requires clients to repeatedly upload full-model updates for aggregation. When the model is large, this process consumes substantial wireless bandwidth and becomes a major systems bottleneck.

\begin{figure*}[t]
\centering
\includegraphics[width=\textwidth]{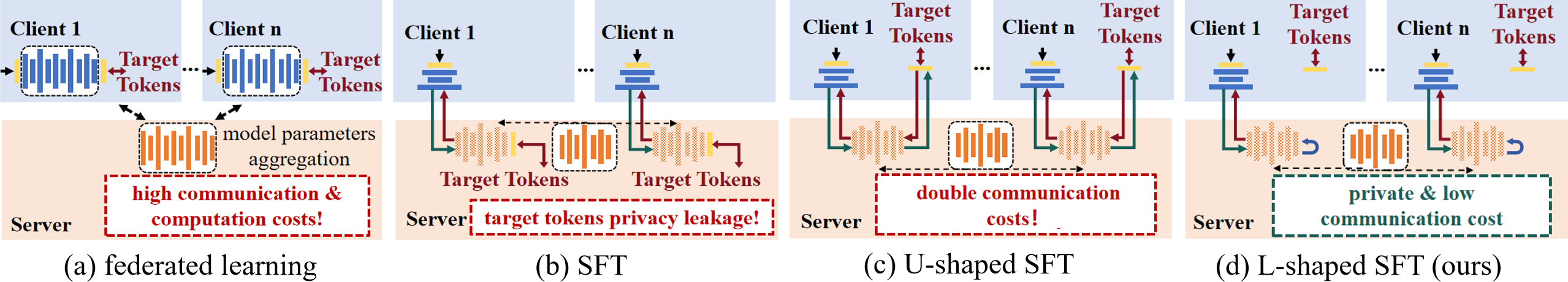}
\vspace{-6mm}
\caption{Evolution of split federated fine-tuning.}
\vspace{-4mm}
\label{intro}
\end{figure*}

These limitations motivate split federated fine-tuning (SFT) \cite{thapa2022splitfed, 11433073}, which partitions the model into a lightweight client-side submodel and a larger server-side submodel. As illustrated in Fig.~\ref{intro}(b), the client executes only the shallow prefix of the model and uploads the resulting cut-layer activations to the server, which performs the remaining forward computation, computes the loss, and returns the corresponding gradients for client-side backpropagation. After several local fine-tuning steps, the client-side submodels are periodically sent to the server and aggregated. By offloading most model parameters and heavy computation to the server, SFT substantially reduces the client-side memory footprint, local computational burden, and the amount of model parameters that must be transmitted during aggregation, making collaborative LLM fine-tuning more feasible for resource-constrained edge devices.
However, this standard SFT design faces a fundamental limitation in autoregressive LLM fine-tuning: the target tokens must be made available at the server for loss computation. To keep target tokens private on device, recent work adopts U-shaped SFT \cite{zhang2025edge,yuan2025flexiblepersonalizedsplitfederated,yang2025fedmobillmefficientfederated}, which moves the prediction head and loss computation back to the client, as illustrated in Fig.~\ref{intro}(c).
Howerver, the U-shaped SFT design faces two main system challenges: \textit{(1) U-shaped SFT doubles communication costs.} Clients must receive the server-side output activations for local loss evaluation and then send the resulting gradients back to the server during backpropagation. Consequently, training becomes tightly coupled to step-by-step bidirectional interaction between the client and the server, which substantially increases communication overhead. \textit{(2) Both standard SFT and U-shaped SFT require the client to remain continuously involved throughout the training process.} In standard SFT, the client must repeatedly participate in split forward and backward propagation to provide activations and receive gradients, while in U-shaped SFT it must further wait for server-side activations and send gradients back for local loss computation. Such tight synchronization is impractical in mobile environments with unstable bandwidth and intermittent connectivity.

Motivated by these challenges, we aim to answer the following research questions:

\textit{(RQ1): How can we redesign the U-shaped SFT pipeline to reduce per-step communication overhead without sacrificing fine-tuning effectiveness?}

\textit{(RQ2): How can we relax the need for continuous client--server synchronization, thereby reducing client online time and improving the practicality of split LLM fine-tuning in unstable edge environments?}

%To keep target tokens on device, recent work adopts U-shaped SFT \cite{zhang2025edge,yuan2025flexiblepersonalizedsplitfederated,yang2025fedmobillmefficientfederated}, which places the prediction head and loss computation back on the client. As shown in Fig.\ref{intro}(c), while this preserves label locality, it introduces a structurally communication redundancy compared with standard SFT: clients must receive the output activtions of the server-side model for local loss evaluation, and then send the calculated gradients back to the server again during backpropagation. As a result, training becomes tightly coupled to step-by-step bidirectional interaction between the clients

In this paper, we present \textbf{L-shaped SFT}, a split fine-tuning framework that removes the bidirectional bottleneck of U-shaped SFT. As shown in Fig.~\ref{intro}(d), the key idea is to redesign where supervision happens in the split pipeline. Instead of returning server-side outputs to the client for loss evaluation, L-shaped SFT exploits weight tying in modern LLMs to supervise server-side hidden activations directly using target embeddings. This moves loss computation to the server side and eliminates the activation-return / gradient-upload loop required by U-shaped designs, thereby substantially reducing per-step communication across the split interface \textit{(answer RQ1)}.

Although our L-shaped SFT effectively reduces communication overhead while keeping target tokens local, it still requires clients to remain continuously involved in the training process. To remove this remaining bottleneck, we further introduce \textbf{one-shot L-shaped SFT} (hereafter referred to as \textbf{one-shot SFT}), an operating mode designed for intermittently connected edge devices. In one-shot SFT, each client uploads its activations only once and then leaves the training loop, while the server continues optimization over cached representations. This design is a natural extension of L-shaped SFT, whose training pipeline eliminates the activation-return/gradient-upload loop and makes server-side training feasible without ongoing client participation. One-shot SFT reduces the communication frequency to a single upload window and minimizes client online time, which makes the fine-tuning process more practical for unstable and resource-constrainted mobile environments \textit{(answer RQ2)}.

Our main contributions are summarized as follows:
\begin{itemize}
\item \textbf{A split-training framework that removes the bidirectional bottleneck.} We present L-shaped SFT, a split federated fine-tuning design that moves supervision from the client to the server and eliminates the activation-return / gradient-upload loop required by U-shaped SFT. This redesign reduces per-step communication across the split interface and weakens the dependence on continuous client participation.

\item \textbf{A one-shot operating mode for unstable edge environments.} We show that the unidirectional structure of L-shaped SFT naturally enables one-shot SFT, where clients upload activations once and then go offline while the server continues optimization over cached representations. This substantially reduces communication frequency and client online time, making split LLM fine-tuning more practical under intermittent connectivity and limited device resources.

\item \textbf{A real-system implementation on heterogeneous edge hardware.} We implement the proposed framework in a real testbed with commercial smartphones, NVIDIA developer boards, and a GPU server. Experiments on WikiText-2 and MMLU show that our design achieves substantial communication savings while maintaining competitive fine-tuning quality in practical deployment settings.
\end{itemize}

\section{Background and Related Work}
\label{sec:related}

\textbf{Split Federated Learning for LLMs.} Split learning and split federated learning have been widely explored as a way to reduce client-side computation and memory by offloading deeper layers to a server. In the context of LLM adaptation, recent efforts such as SplitLoRA~\cite{lin2024splitlora} and SflLLM~\cite{11045165} adopt standard SFT designs for LLM fine-tuning to reduce the amount of trainable state maintained on resource-constrained devices. These studies show that split execution is a promising systems direction for edge LLM fine-tuning, especially when full model is infeasible on devices. However, standard SFT designs implicitly assume that supervision can be computed at the server. In autoregressive LLM fine-tuning, this assumption is problematic because the target tokens may contain sensitive user content and are therefore undesirable to expose outside the client.
\
\\
\noindent \textbf{U-shaped SFT.} To preserve label locality in autoregressive tasks, recent work moves the prediction head and loss computation back to the client, resulting in U-shaped SFT pipelines. Representative examples include EUSFL~\cite{zhang2025edge}, MPSL \cite{fudala2025fine}, FlexP-SFL~\cite{yuan2025flexiblepersonalizedsplitfederated}, and MobiLLM~\cite{yang2025fedmobillmefficientfederated}. These designs ensure that target tokens remain on device, but they also require the server-side outputs to be returned to the client for local loss evaluation, followed by a client-to-server gradient upload during backpropagation.
\
\\
\noindent \textbf{Communication Bottleneck in U-shaped SFT.} Concretely, in U-shaped SFT, the client first uploads cut-layer activations to the server, the server executes the deeper layers, and the resulting server-side activations are then returned to the client for local loss computation. During backpropagation, gradients must be sent back to the server before the server can continue updating its own layers.

\begin{figure}[t]
\centering
\includegraphics[width=0.46\textwidth]{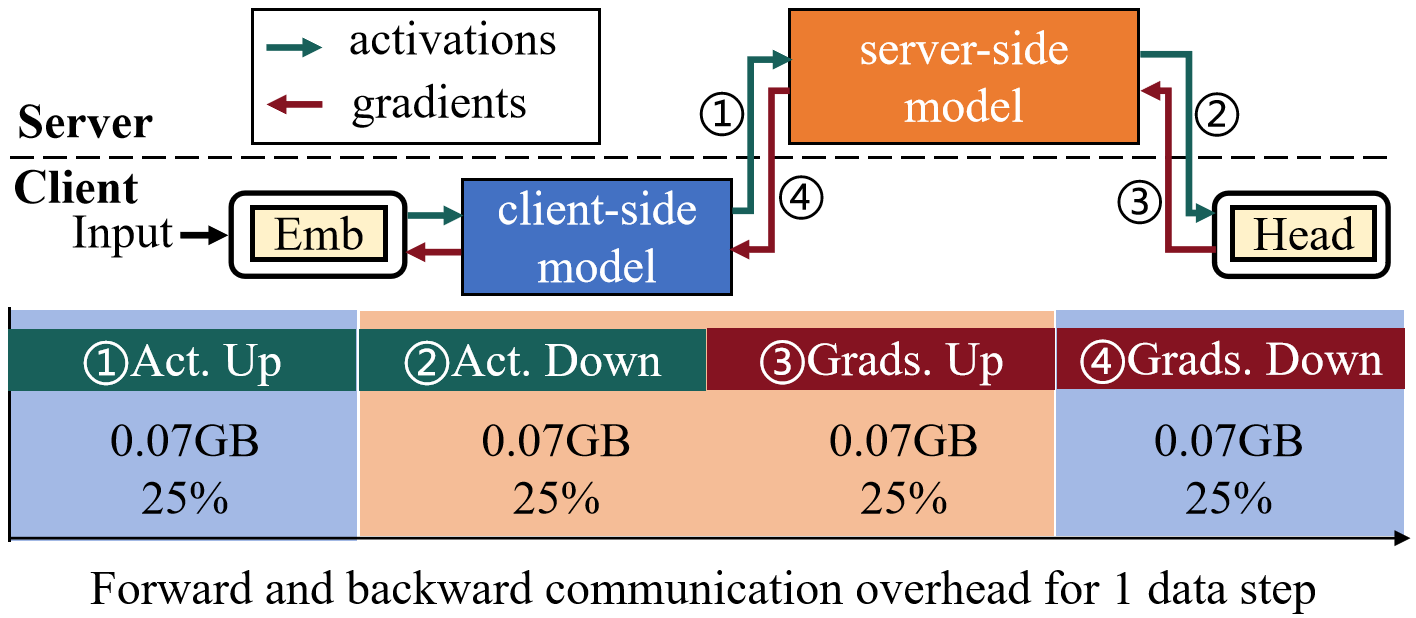}
\vspace{-2mm}
\caption{\textbf{Communication bottleneck analysis of U-shaped SFT.} Unlike standard SFT, U-shaped architectures necessitate a \textbf{bidirectional gradient loop}. (Experimental Setup: Qwen2.5-0.5B on MMLU, sequence length: 128, batch size: 8).}
\label{fig:comm_comparison}
\vspace{-5mm}
\end{figure}

To quantify the systems overhead of U-shaped SFT, we empirically measure its communication cost for a single training step. As shown in Fig.~\ref{fig:comm_comparison}, under Qwen2.5-0.5B on MMLU with sequence length 128 and batch size 8, each step involves four cross-interface transmissions:
\begin{enumerate}
    \item \textbf{activation upload:} Client $\rightarrow$ Server,
    \item \textbf{activation download:} Server $\rightarrow$ Client,
    \item \textbf{gradient upload:} Client $\rightarrow$ Server,
    \item \textbf{gradient download:} Server $\rightarrow$ Client.
\end{enumerate}

By contrast, standard SFT only requires activation upload (step 1) and gradient download (step 4). The additional activation download (step 2) and gradient upload (step 3) path in U-shaped SFT therefore introduces a mandatory bidirectional loop, which substantially increases the communication overhead of each training step.

More importantly, because loss computation remains on the client, the client and server must stay tightly synchronized throughout fine-tuning. In practice, this means that the client must remain continuously online, wait for server responses, and sustain communication with the server at every training step. Such step-by-step bidirectional coupling is undesirable for bandwidth-constrained and intermittently connected edge devices, where connectivity and availability may fluctuate significantly during fine-tuning.
\
\\
\noindent \textbf{Positioning of our work.}
Our work targets a systems bottleneck that remains unresolved in prior SFT designs for autoregressive LLMs. Standard SFT places loss computation on the server, which requires exposing target tokens outside the client. U-shaped SFT avoids this by moving the prediction head and loss computation back to the client, but introduces a bidirectional dependency across the split interface in every training step. This dependency increases per-step communication and requires continuous client participation. Our goal is to improve the practicality of standard SFT and U-shaped SFT under bandwidth-constrained and intermittently connected edge environments while keeping target tokens locally. We redesign the split pipeline, which yields a unidirectional training protocol that reduces communication overhead at each step and naturally enables a one-shot operating mode that minimizes client online time.

\begin{figure}[t]
\centering
\includegraphics[width=0.46\textwidth]{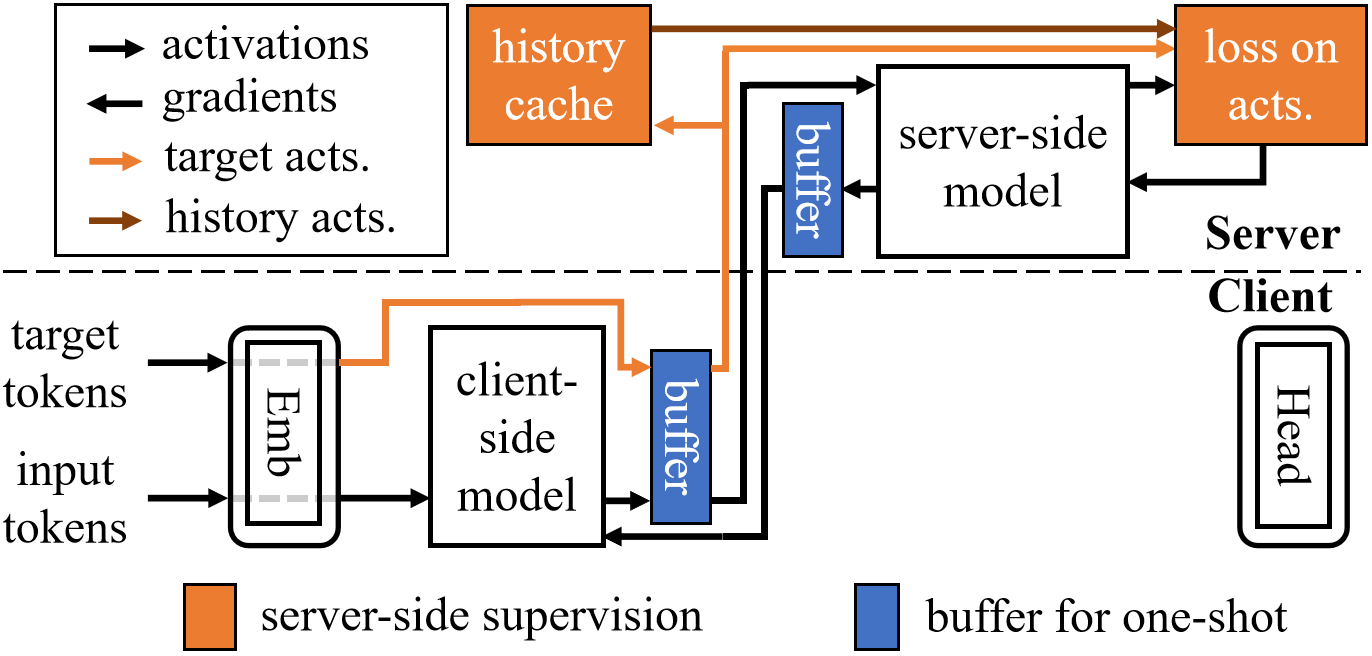}
\vspace{-2mm}
\caption{Workflow of L-shaped SFT.}
\label{system}
\vspace{-5mm}
\end{figure}

\section{System Design}
\label{sec:methodology}

This section presents the design of L-shaped SFT\footnote{In this work, we focus on SFL-v2\cite{thapa2022splitfed}. During split execution, we assume that the server directly accesses only the server-side submodel parameters, while client-side components such as the embedding layer and the output head remain local and are not directly exposed to the server. This is consistent with the standard split learning settings~\cite{thapa2022splitfed}.}. Our goal is to remove the step-by-step bidirectional dependency in U-shaped SFT while keeping target tokens local to the client. To this end, we redesign where supervision is imposed in the split pipeline, so that the server can complete fine-tuning without returning intermediate outputs to the client at every step. The resulting design has two key properties: it reduces per-step communication across the split interface and naturally enables a one-shot operating mode for intermittently connected edge devices.

\subsection{Overview and Setup}
\label{subsec:setup}

We consider split federated fine-tuning for autoregressive LLMs, where the model is partitioned into a client-side submodel $f_c(\cdot;\theta_c)$ and a server-side submodel $f_s(\cdot;\theta_s)$. For an input sequence $x$ with target tokens $y$, the client first computes the cut-layer activation
$
h_c = f_c(x;\theta_c),
$
and uploads it to the server. The server then continues the forward pass and produces
$
h_s = f_s(h_c;\theta_s).
$

In conventional U-shaped SFT, loss computation remains on the client. As a result, the server-side output must be returned to the client for local loss evaluation, and the corresponding gradients must then be sent back to the server during backpropagation. This creates a mandatory bidirectional interaction at every fine-tuning step. In contrast, L-shaped SFT moves supervision to the server side, so that training can proceed without returning server-side outputs to the client. Figure~\ref{system} illustrates the overall workflow.
%\footnote{Within each communication round, participating clients perform multiple step-level split-training updates with the server. At the end of each communication round, the client-side submodels are aggregated to form the updated global client-side model.}.

Throughout this section, we use $e_y=\mathrm{Embed}(y)$ to denote the target embedding of $y$. This target embedding is computed locally on the client and serves as the supervisory signal for server-side fine-tuning.

\begin{figure}[t]
\centering
\includegraphics[width=0.46\textwidth]{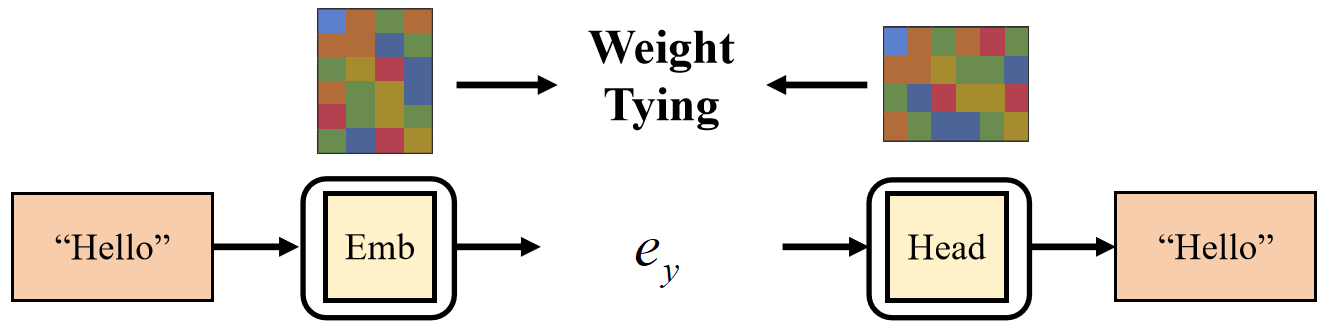}
\vspace{-4mm}
\caption{Weight Tying between the embedding layer and the head layer.}
\label{weight_tying}
\vspace{-5mm}
\end{figure}

\subsection{Server-Side Supervision with Target Embeddings}
\label{subsec:server_supervision}

The main obstacle in redesigning U-shaped SFT is that, for autoregressive fine-tuning, the target tokens are needed at the site of loss computation. Existing U-shaped designs address this by keeping the prediction head and loss computation on the client, but this choice also induces the bidirectional communication loop discussed in Section~\ref{sec:related}. Our key idea is to shift the loss calculation from client side to server side.

This redesign leverages the weight-tying property in modern LLMs \cite{NIPS2017_3f5ee243,press-wolf-2017-using}.\footnote{Our method also applies to models with an untied head, where the client can instead provide only the required head vector. This does not explicitly reveal the target token unless the vector can be matched to a known vocabulary-indexed output matrix.} Since output head and input embedding are typically large, weight tying is widely used to reduce model storage and memory, which is particularly suitable for mobile devices. As shown in Fig.~\ref{weight_tying}, the output head is typically tied to the input embedding matrix, i.e.,
\begin{equation}
W_{\text{head}} = E^\top.
\end{equation}
Although $E^\top E$ is not an identity matrix in a strict algebraic sense, pre-trained embedding spaces exhibit strong semantic consistency: if a hidden state is well aligned with the embedding of the correct token, the tied output head tends to assign that token the highest score.

We empirically verify this property by testing the "self-recovery" rate of embeddings across standard LLMs by projecting embeddings back through the transposed embedding matrix. As shown in Table~\ref{tab:semantic_consistency}, the embedding self-recovery rate is consistently high across representative LLM families. For GPT-2 and Llama-3.2, the Top-1 recovery rate exceeds $99.5\%$, while Qwen2.5 still achieves a Top-1 recovery rate above $94\%$ despite its substantially larger vocabulary. These results support the use of target embeddings as effective supervisory anchors.

Based on this observation, L-shaped SFT replaces client-side token-level supervision with server-side embedding-level supervision. Instead of requiring the client to receive server outputs and compute the loss locally, the client only uploads the target embedding
\begin{equation}
e_y = \mathrm{Embed}(y),
\end{equation}
together with the cut-layer activation $h_c$. The server then compares its predicted hidden state $h_s$ against $e_y$ directly. This redesign removes the activation-return / gradient-upload loop required by U-shaped SFT, thereby reducing per-step communication while avoiding the need to expose raw target token IDs outside the client.

\begin{table}[]
\centering
\caption{\textbf{Empirical Validation of Semantic Consistency.} We evaluate the "self-recovery" accuracy ($y' = \text{argmax}(e_y E^T)$) across diverse model families. The high recovery rates verify that target embeddings serve as precise supervision labels.}
\vspace{-2mm}
\label{tab:semantic_consistency}
\resizebox{0.9\linewidth}{!}{
\begin{tabular}{lccc}
\toprule
\textbf{Model} & \textbf{Vocab Size} & \textbf{Top-1 (\%)} & \textbf{Top-5 (\%)} \\
\midrule
GPT-2 & 50,257 & 99.66 & 99.88 \\
Llama-3.2-1B-Instruct & 128,256 & 99.52 & 99.58 \\
Qwen2.5-0.5B & 151,699 & 94.34 & 94.88 \\
\bottomrule
\end{tabular}
}
\vspace{-6mm}
\end{table}

\subsection{Activation-Space Objective}
\label{subsec:aloss_design}

Once supervision is moved to the server side, the loss can no longer be computed through the standard prediction head on the client. We therefore define the training objective directly in the activation space.

Specifically, we reformulate next-token supervision as a contrastive objective \cite{9578660} between the server-side output activation and the target embedding. Let $\mathcal{Q}$ denote a history cache maintained on the server, which stores target embeddings from previous steps or other clients and serves as a lightweight proxy for negative samples. For a token position $t$, we define the activation-based loss as
\begin{equation}
\begin{split}
\mathcal{L}_{\text{A-Loss}}
&=
-\log
\frac{\exp(s_{y}/\tau)}
{\exp(s_{y}/\tau)+\sum_{e_k\in\mathcal{Q}}\exp(s_{k}/\tau)}, \\
\text{where}\quad
s_{k}
&=
\mathrm{sim}(h_{s}, e_k).
\end{split}
\label{eq:aloss}
\end{equation}
Here, $\tau$ is a temperature hyperparameter, and $\mathrm{sim}(\cdot,\cdot)$ is a similarity function. In practice, the dot product is a natural choice because it matches the form of the unnormalized logits induced by a tied output head, while cosine similarity can also be used when additional normalization is preferred.

Intuitively, $\mathcal{L}_{\text{A-Loss}}$ pulls the predicted hidden state toward the target embedding and pushes it away from other embeddings in the history cache. This allows the server to impose token-level supervision without reconstructing the full vocabulary projection on the client side and without requiring bidirectional interaction across the split interface.

The relationship between $\mathcal{L}_{\text{A-Loss}}$ and standard cross-entropy is analyzed in Section~4. In particular, Section~4 shows that under the weight-tying condition and an appropriate cache setting, the proposed loss recovers the same gradient signal at the hidden-state interface.

Algorithm~\ref{alg:lshaped_sft} summarizes the training workflow of L-shaped SFT. In each step, the client uploads the cut-layer activation together with the target embedding, and the server performs both forward supervision and loss computation in the activation space. Unlike U-shaped SFT, the server-side output does not need to be returned to the client for local head computation. The only backward information sent across the split interface is the cut-layer gradient, which is used by the client to update its local submodel.

\begin{algorithm}[t]
\caption{L-shaped SFT}
\label{alg:lshaped_sft}
\begin{algorithmic}[1]
\Require Client set $\mathcal{C}=\{1,\dots,N\}$; local dataset $\mathcal{D}_i$ for each client $i$; initial global client-side parameters $\bar{\theta}_c^{\,0}$; server-side parameters $\theta_s$; client-side submodel $f_c(\cdot;\theta_c)$; server-side submodel $f_s(\cdot;\theta_s)$; embedding function $\mathrm{Embed}(\cdot)$; server-side history cache $\mathcal{Q}$; communication rounds $R$; local split-training steps per round $T$
\Ensure Aggregated client-side parameters $\bar{\theta}_c^{\,R}$ and updated server-side parameters $\theta_s$

\For{$r = 1$ to $R$}
    \ForAll{participating client $i \in \mathcal{C}$ \textbf{in parallel}}
        \State Initialize local client-side model: $\theta_{c,i}^{\,r,0} \gets \bar{\theta}_c^{\,r-1}$
    \EndFor

    \For{$t = 1$ to $T$}
        \ForAll{participating client $i \in \mathcal{C}$ \textbf{in parallel}}
            \State Sample a mini-batch $(x_i^t, y_i^t)$ from $\mathcal{D}_i$
            \State $h_{c,i}^t \gets f_c(x_i^t;\theta_{c,i}^{\,r,t-1})$
            \State $e_{y,i}^t \gets \mathrm{Embed}(y_i^t)$
            \State Upload $(h_{c,i}^t, e_{y,i}^t)$ to the server
        \EndFor

        \ForAll{received batch $(h_{c,i}^t, e_{y,i}^t)$}
            \State $h_{s,i}^t \gets f_s(h_{c,i}^t;\theta_s)$
            \State Compute $\mathcal{L}_{\text{A-Loss}}(h_{s,i}^t, e_{y,i}^t; \mathcal{Q})$
            \State Backpropagate on the server and update $\theta_s$
            \State Obtain cut-layer gradient $\nabla h_{c,i}^t = \partial \mathcal{L}/\partial h_{c,i}^t$
            \State Update $\mathcal{Q}$ with $e_{y,i}^t$
            \State Send $\nabla h_{c,i}^t$ back to client $i$
        \EndFor

        \ForAll{participating client $i \in \mathcal{C}$ \textbf{in parallel}}
            \State Backpropagate through $f_c(\cdot;\theta_{c,i}^{\,r,t-1})$ using $\nabla h_{c,i}^t$
            \State Update local client-side parameters to $\theta_{c,i}^{\,r,t}$
        \EndFor
    \EndFor

    \State Aggregate client-side models:
    $$
    \bar{\theta}_c^{\,r} \gets \sum_{i \in \mathcal{C}} p_i \theta_{c,i}^{\,r,T},
    $$
    \Statex \hspace{\algorithmicindent} where $p_i$ is the aggregation weight of client $i$
\EndFor

\State \Return $\bar{\theta}_c^{\,R}, \theta_s$
\end{algorithmic}
\end{algorithm}

\vspace{-2mm}
\subsection{One-shot SFT}
\label{subsec:oneshot_design}

The unidirectional structure of L-shaped SFT naturally enables a more aggressive operating mode, which we call \textit{one-shot SFT}. The goal of one-shot SFT is to minimize client online time in unstable edge environments by decoupling data provision from server-side optimization. Algorithm~\ref{alg:oneshot_sft} summarizes the overall workflow.

\textbf{1) Single local forward pass.}
Each client performs a single forward pass over its local dataset. For each training sample $(x,y)$, the client computes the cut-layer activation
$
h_c = f_c(x;\theta_c),
$
and the target embedding
$
e_y = \mathrm{Embed}(y).
$
These pairs $(h_c,e_y)$ are stored in a local buffer.

\textbf{2) One-time transmission.}
After processing its local dataset once, the client uploads the buffered activation--embedding pairs to the server in a single transmission. Unlike step-synchronous U-shaped SFT, this process does not require the client to wait for server responses or remain continuously reachable during later optimization.

\textbf{3) Server-side caching and offline training.}
The server caches all uploaded pairs as a training buffer. It then iterates over this cached buffer for multiple epochs, computes
$
h_s = f_s(h_c;\theta_s),
$
evaluates $\mathcal{L}_{\text{A-Loss}}(h_s,e_y)$, and updates the server-side parameters by backpropagation. After the initial upload, no further communication with the clients is required, and the clients can immediately go offline or release their local resources.

Compared with step-synchronous split training, one-shot SFT changes the role of the edge client from an always-online participant to a transient data provider. The communication pattern is reduced from repeated step-level interaction to a single upload window, while the optimization itself is completed entirely on cached server-side representations. This design is particularly suitable for bandwidth-constrained and intermittently connected edge devices, where maintaining long-lived synchronization with the server is costly or unreliable.

\vspace{-2mm}
\section{Theoretical Analysis}
\label{sec:analysis}

The purpose of this section is to show that the proposed A-Loss is not an ad hoc training objective, but a principled surrogate of the standard cross-entropy loss for autoregressive fine-tuning. Our analysis focuses on the gradient returned by the server-side objective, since this gradient directly determines the optimization signal received by both the server-side submodel and the client-side submodel through the split interface. Relative to prior convergence analyses~\cite{11049940,NEURIPS2024_bb63841e}, the main novelty of our method lies in replacing the standard cross-entropy objective with the proposed A-Loss on the server side. Accordingly, the goal of this section is not to re-derive a full convergence bound, but to establish the consistency between A-Loss and the standard cross-entropy loss. In particular, if A-Loss induces the same gradient as cross-entropy at the server-side output, and hence at the split interface, then under the same optimization assumptions as in prior analyses, our method inherits the same convergence characterization.

Specifically, we establish two results. First, under ideal conditions, A-Loss yields exactly the same gradient with respect to the server-side hidden state as the standard cross-entropy objective. This implies that the proposed server-side supervision preserves the training signal of conventional autoregressive fine-tuning. Second, when the history cache is finite, A-Loss becomes a cache-restricted approximation to full-vocabulary cross-entropy, and the resulting gradient deviation can be explicitly characterized. Together, these results justify why A-Loss can serve as a practical replacement for client-side loss computation in L-shaped SFT.

\begin{algorithm}[t]
\caption{One-shot SFT}
\label{alg:oneshot_sft}
\begin{algorithmic}[1]
\Require Client set $\mathcal{C}=\{1,\dots,N\}$; local dataset $\mathcal{D}_i$ for each client $i$; client-side submodel $f_c(\cdot;\theta_c)$; server-side submodel $f_s(\cdot;\theta_s)$; embedding function $\mathrm{Embed}(\cdot)$; server-side training epochs $E$
\Ensure Updated server-side parameters $\theta_s$

\State $\mathcal{B} \gets \emptyset$ \Comment{Global buffer stored on the server}

\ForAll{client $i \in \mathcal{C}$ \textbf{in parallel}}
    \State $\mathcal{B}_i \gets \emptyset$  \Comment{Local buffer on clients}
    \ForAll{$(x,y) \in \mathcal{D}_i$}
        \State $h_c \gets f_c(x;\theta_c)$
        \State $e_y \gets \mathrm{Embed}(y)$
        \State $\mathcal{B}_i \gets \mathcal{B}_i \cup \{(h_c,e_y)\}$
    \EndFor
    \State Upload $\mathcal{B}_i$ to the server once
\EndFor

\State $\mathcal{Q} \gets \mathcal{B} \gets \bigcup_{i=1}^{N} \mathcal{B}_i$ \Comment{Aggregate local buffers and update history cache}
\State Clients go offline

\For{$epoch = 1$ to $E$}
    \ForAll{$(h_c,e_y) \in \mathcal{B}$}
        \State $h_s \gets f_s(h_c;\theta_s)$
        \State Compute $\mathcal{L}_{\text{A-Loss}}(h_s,e_y)$
        \State Update $\theta_s$ by backpropagation
    \EndFor
\EndFor

\State \Return $\theta_s$
\end{algorithmic}
\end{algorithm}

\vspace{-2mm}
\subsection{Gradient Consistency of A-Loss}
\label{subsec:grad_consistency}

We begin by comparing the gradient induced by A-Loss with that of the standard cross-entropy loss. Let $h_s$ denote the server-side hidden state produced by the server-side submodel. In standard autoregressive fine-tuning, the next-token distribution is computed by a linear output head $W_{\text{head}} \in \mathbb{R}^{V \times d}$ followed by Softmax, where $V$ is the vocabulary size and $d$ is the hidden dimension. Let $\mathbf{w}_k$ denote the embedding vector associated with token $k$ in the output head, and let $y$ denote the ground-truth token. Then the cross-entropy loss induces the following gradient with respect to $h_s$:
\begin{equation}
\label{eq:ce_grad_compact}
\nabla_{h_s}\mathcal{L}_{\mathrm{CE}}
=
\sum_{k=1}^{V} p_k \mathbf{w}_k - \mathbf{w}_y,
\end{equation}
where $p_k$ is the Softmax probability assigned to token $k$.

For the proposed A-Loss, the server does not predict logits over the full vocabulary. Instead, it compares the hidden state $h_s$ with the target embedding and a set of cached negative embeddings. Let $\mathcal{Q}$ denote the history cache, let $\mathbf{e}_k$ be the embedding of token $k$, and let $q_k$ be the cache-normalized contrastive probability under A-Loss. Using dot-product similarity, the gradient of A-Loss with respect to $h_s$ can be written as
\begin{equation}
\label{eq:aloss_grad_compact}
\nabla_{h_s}\mathcal{L}_{\mathrm{A\mbox{-}Loss}}
=
\frac{1}{\tau}
\left(
\sum_{k\in\mathcal{Q}} q_k \mathbf{e}_k - \mathbf{e}_y
\right),
\end{equation}
where $\tau$ is the temperature coefficient.

The key question is whether this gradient is consistent with the gradient produced by standard cross-entropy. The following theorem gives the answer.

\begin{theorem}
\label{thm:grad_equiv_rewrite}
If the following three conditions hold: (i) the temperature is set to $\tau=1$; (ii) the history cache covers the full vocabulary, i.e., $\mathcal{Q}=\{1,\dots,V\}$; and (iii) weight tying holds so that $\mathbf{e}_k=\mathbf{w}_k$ for all tokens $k$, then
\begin{equation}
\label{eq:grad_equiv_main}
\nabla_{h_s}\mathcal{L}_{\mathrm{A\mbox{-}Loss}}
=
\nabla_{h_s}\mathcal{L}_{\mathrm{CE}}.
\end{equation}
Consequently, the two objectives also induce the same gradient at the split interface.
\end{theorem}

\noindent\textbf{Proof intuition.}
Under weight tying, the similarity score used by A-Loss is computed against the same token embeddings as the standard output head. When $\tau=1$ and the cache covers the full vocabulary, the normalization term of A-Loss becomes identical to the Softmax partition function. As a result, the cache-normalized probability $q_k$ reduces to the standard Softmax probability $p_k$, and the target embedding $\mathbf{e}_y$ coincides with the output-head vector $\mathbf{w}_y$. Substituting these relations into Eq.~\eqref{eq:aloss_grad_compact} immediately recovers Eq.~\eqref{eq:ce_grad_compact}, which proves Eq.~\eqref{eq:grad_equiv_main}.

Theorem~\ref{thm:grad_equiv_rewrite} shows that A-Loss preserves the same first-order optimization signal as standard cross-entropy, even though the loss is no longer computed through a client-side prediction head. In other words, it can be exactly equivalent to conventional token-level supervision from the viewpoint of gradient-based optimization.

This equivalence also extends naturally to split training. Since the gradient returned to the client is obtained by backpropagating through the server-side submodel, the cut-layer gradient is
\begin{equation}
\label{eq:split_grad_compact}
\nabla_{h_c}\mathcal{L}
=
\left(
\frac{\partial h_s}{\partial h_c}
\right)^\top
\nabla_{h_s}\mathcal{L}.
\end{equation}
Therefore, once Eq.~\eqref{eq:grad_equiv_main} holds at the server-side hidden state, the same equality also holds for the split-interface gradient returned to the client. This explains why L-shaped SFT can replace the U-shaped training loop without changing the underlying training signal in the ideal case.

\subsection{Approximation Error under Finite Cache}
\label{subsec:finite_cache}

In practice, the server-side history cache only contains a subset of vocabulary items, rather than the entire vocabulary. In this case, A-Loss should be interpreted as a cache-restricted approximation to the standard cross-entropy objective. The exact equivalence in Theorem~\ref{thm:grad_equiv_rewrite} no longer holds, but the deviation can be characterized in a compact form.

Assume that $\tau=1$, weight tying still holds, and the positive token $y$ is included in the cache $\mathcal{Q}$. We define the gradient gap at the server-side hidden state as
\begin{equation}
\label{eq:finite_cache_gap_compact}
\Delta_{h_s}
\triangleq
\nabla_{h_s}\mathcal{L}_{\mathrm{A\mbox{-}Loss}}
-
\nabla_{h_s}\mathcal{L}_{\mathrm{CE}}.
\end{equation}
By comparing Eq.~\eqref{eq:ce_grad_compact} and Eq.~\eqref{eq:aloss_grad_compact}, this gap can be decomposed as
\begin{equation}
\label{eq:finite_cache_gap_decomp}
\Delta_{h_s}
=
\sum_{k\in\mathcal{Q}} (q_k-p_k)\mathbf{e}_k
-
\sum_{k\notin\mathcal{Q}} p_k \mathbf{e}_k.
\end{equation}

Equation~\eqref{eq:finite_cache_gap_decomp} reveals two distinct sources of approximation error. The first term,
$
\sum_{k\in\mathcal{Q}} (q_k-p_k)\mathbf{e}_k,
$
is a \emph{renormalization error}: it arises because A-Loss normalizes probabilities only over cached tokens, whereas cross-entropy normalizes over the entire vocabulary. The second term,
$
-\sum_{k\notin\mathcal{Q}} p_k \mathbf{e}_k,
$
is a \emph{missing-negative error}: it captures the contribution of vocabulary items that are absent from the cache but still carry non-negligible probability mass under the full Softmax.

This decomposition provides a clear interpretation of when A-Loss works well in practice. If the cache contains the dominant competing tokens, then the omitted probability mass outside $\mathcal{Q}$ is small, and the renormalized cache distribution $q_k$ remains close to the full Softmax distribution $p_k$. In this case, the gradient gap $\Delta_{h_s}$ becomes small, and A-Loss provides a good approximation to standard cross-entropy. In the limiting case where $\mathcal{Q}$ covers the full vocabulary, both error terms vanish and the exact equivalence in Theorem~\ref{thm:grad_equiv_rewrite} is recovered.

The same approximation error also propagates to the split interface. Let
\begin{equation}
\label{eq:finite_cache_split_gap}
\Delta_{h_c}
\triangleq
\nabla_{h_c}\mathcal{L}_{\mathrm{A\mbox{-}Loss}}
-
\nabla_{h_c}\mathcal{L}_{\mathrm{CE}}.
\end{equation}
Then, by the chain rule,
\begin{equation}
\label{eq:finite_cache_split_gap_chain}
\Delta_{h_c}
=
\left(
\frac{\partial h_s}{\partial h_c}
\right)^\top
\Delta_{h_s}.
\end{equation}

Therefore, under a finite history cache, the quality of the split-interface gradient returned to the client is fully determined by how well the cache approximates the dominant competitors in the full-vocabulary Softmax. This result is important for our system design: it shows that L-shaped SFT does not require exact vocabulary-level supervision at every step. Instead, as long as the history cache captures sufficiently informative negatives, the server can provide a gradient signal that remains close to that of standard autoregressive fine-tuning, while avoiding the bidirectional communication loop of U-shaped SFT.\footnote{As for the privacy of uploaded target embeddings, the client can further protect them using a client-private reversible transformation. Specifically, the client applies a private transformation matrix \(A\) and uploads only transformed embeddings. Since \(A\) is kept locally, the server cannot directly recover embeddings or tokens using the public output head. When needed, the client applies \(A^T\) to restore the original embedding.}

\section{Implementation}

\subsection{Real-world Testbed}
We deploy our framework on a \textbf{real-world wireless testbed} with heterogeneous edge clients. The server is a Dell PowerEdge T640 equipped with two NVIDIA A800 GPUs. The client side consists of five physical devices: 

(i) \textbf{two NVIDIA Jetson Orin Nano 4GB developer boards:} Each Jetson Orin Nano 4GB is built on the NVIDIA Orin SoC and provides up to 20 INT8 TOPS of AI compute, with a 512-core NVIDIA Ampere GPU, 16 Tensor Cores, a 6-core Arm Cortex-A78AE CPU, and 4 GB LPDDR5 memory; on each board, we further configure 8 GB swap memory.  

(ii) \textbf{three Huawei nova 9 Pro smartphones:} Each Huawei nova 9 Pro is equipped with 8 GB RAM and a Qualcomm Snapdragon 778G mobile chipset. 

The client-server communication is carried over a \textbf{Wi-Fi 6} network. 

For the software, model fine-tuning is implemented using \textit{PyTorch} and \textit{Transformers} \cite{DBLP:journals/corr/abs-1910-03771} on the server. The overall federated orchestration is managed by \textbf{\textit{Flower}}\cite{gao2025flowertune}, and the communication layer is implemented with the \textit{Flower C++ SDK}. On the client side, the smartphones run \textbf{\textit{MobileFineTuner}}\cite{geng2025mobilefinetunerunifiedendtoendframework,geng2026edgeflowertuneevaluatingfederatedllm}, a C++ framework for fine-tuning LLMs on commercial Android phones, while the Jetson devices use \textit{PyTorch} and \textit{Transformers}. Communication overhead is measured using Wireshark by recording the actual traffic exchanged between clients and the server during fine-tuning. The full system implementation comprises 17,412 lines of code. Our code is open-sourced at \url{https://github.com/BESTTOOLBOX/Just-Talk-Once}.

\begin{table*}[t]
\centering
\caption{Comparison with U-shaped SFT under \textbf{full fine-tuning}. For MMLU, we report Accuracy ($\uparrow$); for WikiText-2, we report Perplexity (PPL, $\downarrow$).}
\vspace{-3mm}
\label{tab:full_finetune_dual_task}
\resizebox{0.9\textwidth}{!}{%
\begin{tabular}{l|l|c|c|cc|cc|cc}
\toprule
\multirow{2}{*}{\textbf{Task}} & \multirow{2}{*}{\textbf{Model}} & \multirow{2}{*}{\textbf{Size}} & \multirow{2}{*}{\textbf{Pretrained}} & \multicolumn{2}{c|}{\textbf{Finetuned Performance}} & \multicolumn{2}{c|}{\textbf{Comm. Cost on Act. and Grad. (GB)}} & \multicolumn{2}{c}{\textbf{Overall Latency (min)}} \\
& & & & \textbf{U-shaped SFT} & \textbf{L-shaped SFT} & \textbf{U-shaped SFT} & \textbf{L-shaped SFT (Red.)} & \textbf{U-shaped SFT} & \textbf{L-shaped SFT (Red.)} \\
\midrule
\multirow{4}{*}{\shortstack{\textbf{MMLU}\\(Acc \% $\uparrow$)}} 
& Gemma3-270M   & 270M & 25.00 & 27.78 & 27.59 & 1.61 & \textbf{0.61 ($\downarrow$61.92\%)} & 190.42 & \textbf{79.93 ($\downarrow$58.02\%)} \\
& GPT2-Medium   & 355M & 21.63 & 25.61 & 25.60 & 1.22 & \textbf{0.61 ($\downarrow$49.80\%)} & 180.29 & \textbf{79.41 ($\downarrow$55.95\%)} \\
& Qwen2.5-0.5B  & 0.5B & 41.87 & 45.63 & 45.24 & 1.71 & \textbf{1.29 ($\downarrow$24.71\%)} & 252.69 & \textbf{167.16 ($\downarrow$33.85\%)} \\
& GPT2-Large    & 744M & 23.02 & 26.19 & 26.10 & 2.93 & \textbf{0.98 ($\downarrow$66.54\%)} & 432.92 & \textbf{127.20 ($\downarrow$70.62\%)} \\
\midrule
\multirow{4}{*}{\shortstack{\textbf{WikiText-2}\\(PPL $\downarrow$)}} 
& Gemma3-270M   & 270M & 7200.86 & 34.28 & 35.06 & 1.83 & \textbf{1.37 ($\downarrow$25.01\%)} & 271.26 & \textbf{178.29 ($\downarrow$34.27\%)} \\
& GPT2-Medium   & 355M & 48.80 & 35.56 & 35.45 & 2.93 & \textbf{2.19 ($\downarrow$25.02\%)} & 432.92 & \textbf{284.22 ($\downarrow$34.35\%)} \\
& Qwen2.5-0.5B  & 0.5B & 30.17 & 16.66 & 16.63 & 2.56 & \textbf{1.92 ($\downarrow$24.97\%)} & 379.03 & \textbf{248.91 ($\downarrow$34.33\%)} \\
& GPT2-Large    & 744M & 43.63 & 23.04 & 23.14 & 3.66 & \textbf{2.75 ($\downarrow$24.99\%)} & 540.87 & \textbf{355.02 ($\downarrow$34.36\%)} \\
\bottomrule
\end{tabular}%
}
\vspace{-3mm}
\end{table*}

\begin{table*}[t]
\centering
\caption{Comparison with U-shaped SFT under \textbf{LoRA fine-tuning}. For MMLU, we report Accuracy ($\uparrow$); for WikiText-2, we report Perplexity (PPL, $\downarrow$).}
\vspace{-3mm}
\label{tab:lora_finetune_dual_task}
\resizebox{0.9\textwidth}{!}{%
\begin{tabular}{l|l|c|c|cc|cc|cc}
\toprule
\multirow{2}{*}{\textbf{Task}} & \multirow{2}{*}{\textbf{Model}} & \multirow{2}{*}{\textbf{Size}} & \multirow{2}{*}{\textbf{Pretrained}} & \multicolumn{2}{c|}{\textbf{Finetuned Performance}} & \multicolumn{2}{c|}{\textbf{Comm. Cost on Act. and Grad.(GB)}} & \multicolumn{2}{c}{\textbf{Overall Latency (min)}} \\
& & & & \textbf{U-shaped SFT} & \textbf{L-shaped SFT} & \textbf{U-shaped SFT} & \textbf{L-shaped SFT (Red.)} & \textbf{U-shaped SFT} & \textbf{L-shaped SFT (Red.)} \\
\midrule
\multirow{7}{*}{\shortstack{\textbf{MMLU}\\(Acc \% $\uparrow$)}} 
& Gemma3-1B    & 1B   & 26.39 & 49.56 & 50.75 & 2.19 & \textbf{1.65 ($\downarrow$24.71\%)} & 324.54 & \textbf{214.43 ($\downarrow$33.93\%)} \\
& Llama3.2-1B  & 1B   & 44.44 & 49.05 & 49.00 & 2.73 & \textbf{0.59 ($\downarrow$78.49\%)} & 403.80 & \textbf{76.23 ($\downarrow$81.12\%)} \\
& Qwen3-1.7B   & 1.7B & 49.60 & 58.93 & 58.35 & 2.93 & \textbf{2.15 ($\downarrow$26.62\%)} & 432.64 & \textbf{317.61 ($\downarrow$26.59\%)} \\
& Llama3.2-3B  & 3B   & 49.65 & 52.97 & 52.01 & 5.86 & \textbf{1.77 ($\downarrow$69.88\%)} & 865.23 & \textbf{228.66 ($\downarrow$73.57\%)} \\
& Gemma3-4B    & 4B   & 34.80 & 51.87 & 50.78 & 2.44 & \textbf{1.84 ($\downarrow$24.71\%)} & 360.52 & \textbf{238.20 ($\downarrow$40.26\%)} \\
& Qwen3-4B     & 4B   & 65.67 & 74.49 & 74.06 & 2.44 & \textbf{1.45 ($\downarrow$40.56\%)} & 364.27 & \textbf{217.60 ($\downarrow$39.64\%)} \\
& Llama3.1-8B  & 8B   & 58.13 & 72.70 & 72.51 & 5.86 & \textbf{2.88 ($\downarrow$50.81\%)} & 865.20 & \textbf{462.17 ($\downarrow$46.58\%)} \\
\midrule
\multirow{7}{*}{\shortstack{\textbf{WikiText-2}\\(PPL $\downarrow$)}} 
& Gemma3-1B    & 1B   & 676.42 & 22.29 & 22.62 & 3.29 & \textbf{2.47 ($\downarrow$25.00\%)} & 486.80 & \textbf{319.53 ($\downarrow$34.36\%)} \\
& Llama3.2-1B  & 1B   & 25.53 & 14.69 & 15.36 & 5.86 & \textbf{0.73 ($\downarrow$87.50\%)} & 865.29 & \textbf{94.66 ($\downarrow$89.06\%)} \\
& Qwen3-1.7B   & 1.7B & 37.57 & 15.52 & 15.55 & 5.87 & \textbf{2.19 ($\downarrow$62.50\%)} & 869.48 & \textbf{283.99 ($\downarrow$67.34\%)} \\
& Llama3.2-3B  & 3B   & 20.52 & 15.55 & 12.57 & 8.79 & \textbf{1.09 ($\downarrow$87.50\%)} & 1297.85 & \textbf{141.99 ($\downarrow$89.06\%)} \\
& Gemma3-4B    & 4B   & 2796.21 & 13.28 & 15.94 & 8.54 & \textbf{3.66 ($\downarrow$57.14\%)} & 1261.83 & \textbf{473.30 ($\downarrow$62.49\%)} \\
& Qwen3-4B     & 4B   & 29.90 & 14.26 & 14.19 & 7.32 & \textbf{1.83 ($\downarrow$75.00\%)} & 1081.57 & \textbf{236.65 ($\downarrow$78.12\%)} \\
& Llama3.1-8B  & 8B   & 21.00 & 12.37 & 12.25 & 11.72 & \textbf{8.79 ($\downarrow$25.00\%)} & 1730.40 & \textbf{1135.87 ($\downarrow$34.36\%)} \\
\bottomrule
\end{tabular}%
}
\vspace{-4mm}
\end{table*}

\subsection{Data Tasks}

We selected two common LLM benchmark tasks for testing.

\begin{itemize}
\item \textbf{Text Generation Task}: We use the WikiText-2 dataset \cite{Devlin2019BERTPO}. We fine-tune the large models on this dataset and observe their perplexity (PPL), which measures the model’s ability to predict the next word in a sequence. To simulate a federated environment, we simply partition the dataset into 5 distinct subsets, assigning one to each client. 
Given the topical clustering of Wikipedia articles, this naive partitioning strategy inherently yields a non-IID data distribution.

\item \textbf{Text-based Question Answering Task}: We use the MMLU (Massive Multitask Language Understanding) dataset \cite{hendrycks2020measuring}. For this task, we input both the question and answer options into the model, and the model generates a response. The predicted answer is then extracted from the model's output, and the model’s accuracy is evaluated by comparing the predicted answer to the correct answer. Crucially, to ensure rigorous evaluation and prevent test-set contamination, we strictly utilize the 'Train' split for fine-tuning and the 'Test' split for evaluation.
To simulate a realistic non-IID federated environment, we first shuffle the training set and then partition it among 5 clients using a Dirichlet distribution \cite{hsu2019measuring} ($\alpha=0.8$) based on the answer choices.

\end{itemize}

\begin{figure}[t]
\centering
\includegraphics[width=0.4\textwidth]{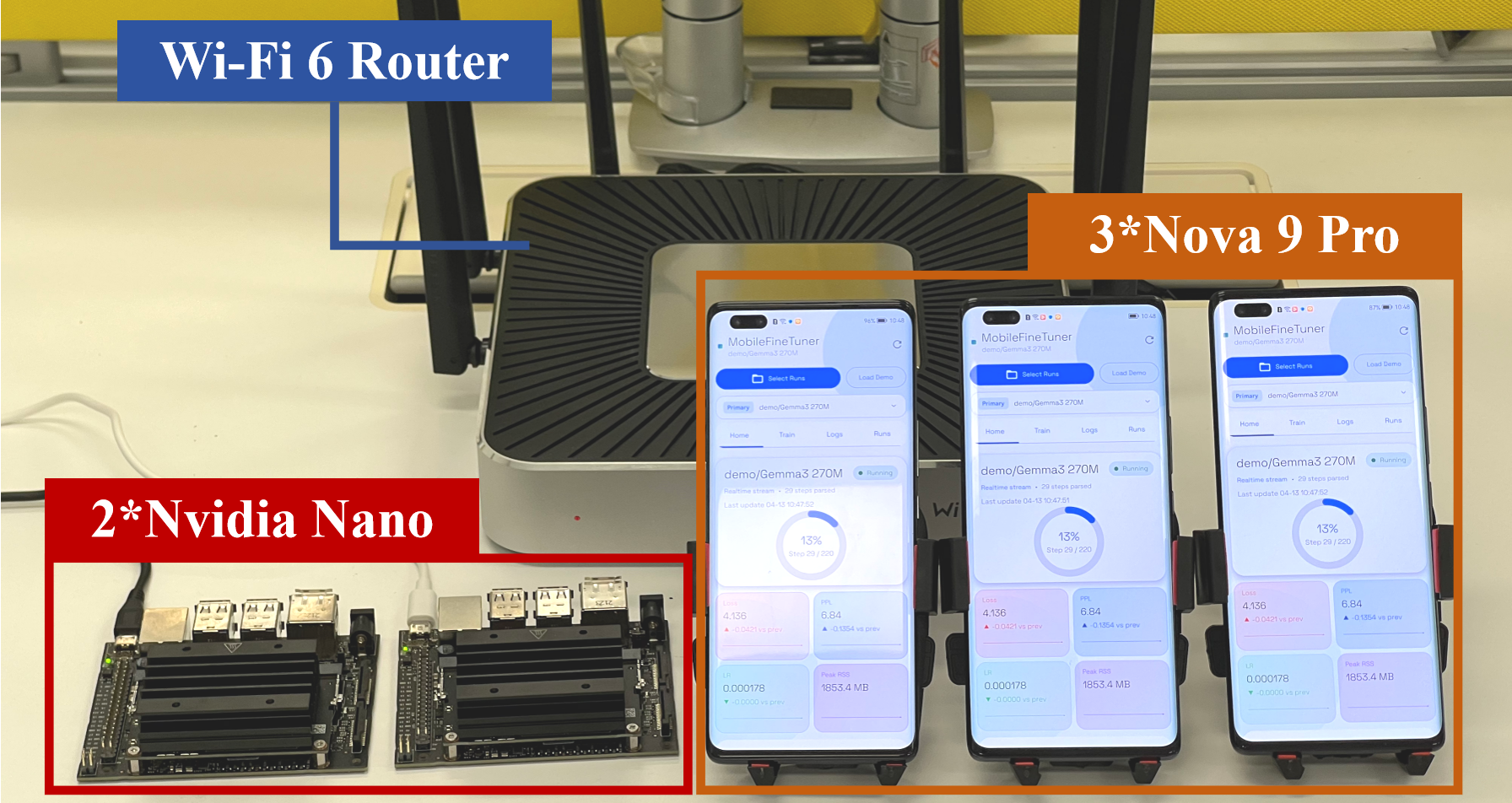}
\vspace{-3mm}
\caption{Our Real-world Testbed: 2*Nvidia Nano, 3*Nova 9 Pro, and Wi-Fi 6 Router.}
\label{test_bed}
\vspace{-5mm}
\end{figure}

\vspace{-2mm}
\subsection{Models}

In our experiments, we perform two types of fine-tuning: full fine-tuning and LoRA fine-tuning. The models used in these experiments vary in size, and we have selected them based on the memory requirements and performance needs of different approaches.

\subsubsection{Full Fine-Tuning}

Full fine-tuning involves training all model parameters, which can be highly memory-intensive. To address the memory constraints, especially for mobile devices, we chose smaller-scale language models with fewer than 1 billion parameters. These models are optimized for efficiency while maintaining competitive performance in various natural language processing (NLP) tasks.

The models selected for full fine-tuning include: GPT-2 Medium (355M), GPT-2 Large (774M), Qwen2.5-0.5B and Gamma3-270M. These models are ideal for mobile environments due to their relatively small size and lower memory consumption and can be used for full fine-tuning.

\subsubsection{LoRA Fine-Tuning}

LoRA fine-tuning enables more efficient adaptation of larger models by focusing on low-rank updates rather than full parameter tuning \cite{hu2022lora}. This technique is especially beneficial when working with larger models, as it reduces the computational and memory overhead.

For LoRA fine-tuning, we selected larger models with more than 1 billion parameters to leverage their advanced capabilities in generating and understanding complex language patterns. These models include: Qwen3-1.7B, Qwen3-4B, Llama3.2-1B, Llama3.2-3B, Llama3.1-8B, Gamma3-1B, Gamma3-4B.

\vspace{-2mm}
\section{Results}
In this section, we report the experimental results of the proposed L-shaped SFT. Our primary baseline is the U-shaped SFT framework. We emphasize that U-shaped SFT is the widely adopted split-training framework used in prior systems such as EUSFL~\cite{zhang2025edge}, MPSL~\cite{fudala2025fine}, FlexP-SFL~\cite{yuan2025flexiblepersonalizedsplitfederated}, and MobiLLM~\cite{yang2025fedmobillmefficientfederated}. For a fair framework-level comparison, we remove the method-specific auxiliary designs in these systems and retain only their shared original U-shaped pipeline as the baseline. This is because the main difference of our method lies in the split-training framework itself, rather than in any additional optimization module. Therefore, the results in this section are designed to isolate the effect of replacing the conventional bidirectional U-shaped pipeline with our L-shaped one. Since our framework redesign is orthogonal to those extra algorithmic components, it can in principle be combined with them.

\vspace{-2mm}
\subsection{Overall Comparison with U-shaped SFT}
\subsubsection{L-shaped SFT}
We first compare the proposed L-shaped SFT with the conventional U-shaped SFT framework. In this group of experiments, both methods use the same split setting, where each client holds the first two hidden layers. Both methods communicate at every local training step and aggregate client-side parameters once every 20 steps. The batch size is 8 and the sequence length is 128. We emphasize that the communication cost reported in Tables~\ref{tab:full_finetune_dual_task} and \ref{tab:lora_finetune_dual_task} only counts the transmission of \textit{activations and gradients} across the split interface, and excludes periodic parameter aggregation, since that part is identical for U-shaped SFT and L-shaped SFT.
\
\\
\noindent\textbf{Full Fine-Tuning.} 
L-shaped SFT achieves performance that is highly comparable to U-shaped SFT across both MMLU and WikiText-2, while both methods consistently outperform the corresponding pretrained models. At the same time, L-shaped SFT substantially reduces communication and end-to-end latency. On MMLU, the communication reduction ranges from 25\% to 66\%, and the latency reduction ranges from 34\% to 70\%. On WikiText-2, the communication reduction is consistently around 25\%, with latency reduction around 34\%.

\noindent\textbf{LoRA Fine-Tuning.}
Across the broader set of 1B-8B models, L-shaped SFT remains competitive with U-shaped SFT in final quality while continuing to reduce communication and overall latency. In particular, on several settings such as Llama3.2-1B and Llama3.2-3B, the communication and latency savings exceed 70\%, and on WikiText-2 the reductions reach as high as 87\% in communication and 89\% in latency. These results confirm that the main advantage of L-shaped SFT lies in improving system efficiency while preserving competitive fine-tuning quality.

\vspace{-2mm}
\subsubsection{One-shot SFT}
We next evaluate one-shot SFT under full fine-tuning setting using Qwen2.5-0.5B. For U-shaped SFT and L-shaped SFT, the client still keeps the first two hidden layers. In contrast, for one-shot SFT, the client performs only a one-time local forward pass and then uploads the required target embeddings to the server, after which the server continues training without further client participation. Table~\ref{tab:oneshot_compare} reports the corresponding system-level comparison. Different from Tables~\ref{tab:full_finetune_dual_task} and \ref{tab:lora_finetune_dual_task}, where the communication cost only counts activation and gradient transmission across the split interface, the communication cost reported here is the overall communication cost, including client-side model parameter upload and download. 

As shown in Table~\ref{tab:oneshot_compare}, one-shot SFT achieves the lowest communication cost, the lowest overall latency, and the shortest client active time. While L-shaped SFT already removes the bidirectional bottleneck of U-shaped SFT, one-shot SFT goes one step further: the client only needs to upload the embeddings corresponding to its local data to the server once, after which the client can directly remain idle.

\subsubsection{Discussion on L-shaped and one-shot SFT}
As shown in Fig.~\ref{fig:mmlu_comm} and Fig.~\ref{fig:mmlu_latency}, both L-shaped SFT and one-shot SFT become increasingly communication- and latency-efficient as the number of local data traversals grows\footnote{Here, ``traversals'' denote the number of times the fine-tuning process passes over the local data. They are used only to analyze how the reduction of L-shaped SFT and one-shot SFT changes during training, and should not be confused with communication rounds.}. In particular, the overall communication reduction achieved by L-shaped SFT keeps increasing with more repeated passes over the same local data. This is because the target embeddings in L-shaped SFT only need to be transmitted once when the local data is first processed, after which the server can reuse the cached embeddings in later traversals without requiring repeated transmission. This explains why many reduction values reported in Section~6.1.1 are already larger than 25\%. One-shot SFT is even more efficient, since it only requires a single communication stage and then completes the remaining optimization entirely on the server side. Therefore, as local data is revisited multiple times, one-shot SFT maintains the largest reduction.

\begin{table}[t]
\centering
\caption{System-level comparison of \textbf{U-shaped SFT}, \textbf{L-shaped SFT}, and \textbf{one-shot SFT}. (Full Fine-tuning on Qwen2.5-0.5B)}
\vspace{-2mm}
\label{tab:oneshot_compare}
\resizebox{0.45\textwidth}{!}{
\begin{tabular}{l|l|c|c|c}
\toprule
\textbf{Task} & \textbf{Method} & \shortstack{\textbf{Overall Comm.}\\\textbf{Cost (GB)}} & \shortstack{\textbf{Overall}\\\textbf{Latency (min)}} & \shortstack{\textbf{Client Active}\\\textbf{Time (min)}} \\
\midrule
\multirow{3}{*}{MMLU}
& U-shaped & 4.79 & 252.69 & 240.66 (95\%) \\
& L-shaped & 4.37 & 167.16 & 133.78 (80\%) \\
& One-shot & 0.98 & 45.57 & 4.83 (11\%)\\
\midrule
\multirow{3}{*}{WikiText-2}
& U-shaped & 6.26 & 379.03 & 362.45 (96\%) \\
& L-shaped & 5.61 & 248.91 & 196.52 (77\%) \\
& One-shot & 1.15 & 48.91 & 6.13 (13\%)\\
\bottomrule
\end{tabular}
}
\vspace{-2mm}
\end{table}

\begin{figure}[t]
\centering
\begin{subfigure}[t]{0.23\textwidth}
    \centering
    \includegraphics[width=\textwidth]{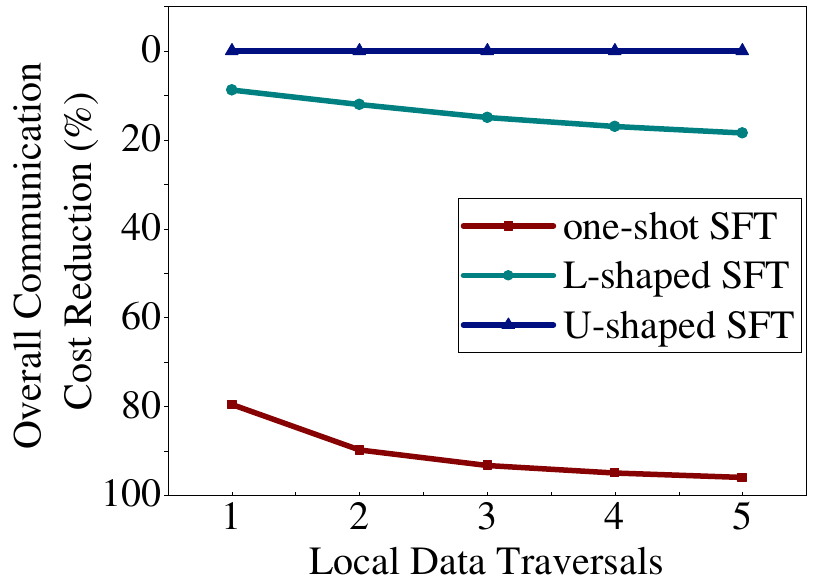}
    \vspace{-6mm}
    \caption{Communication reduction.}
    \label{fig:mmlu_comm}
\end{subfigure}
\hfill
\begin{subfigure}[t]{0.22\textwidth}
    \centering
    \includegraphics[width=\textwidth]{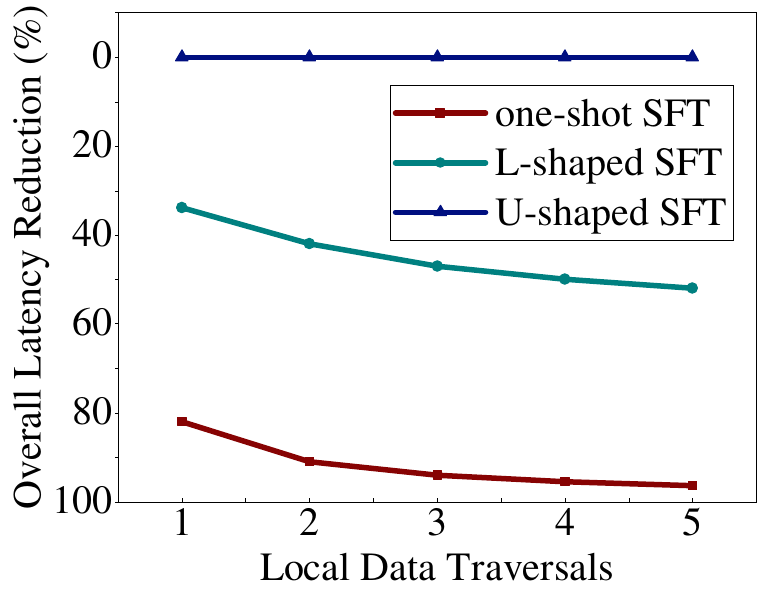}
    \vspace{-6mm}
    \caption{Latency reduction.}
    \label{fig:mmlu_latency}
\end{subfigure}
\vspace{-4mm}
\caption{Reduction versus the number of local data traversals on MMLU with Qwen2.5-0.5B under full fine-tuning.}
\label{fig:mmlu_reduction}
\vspace{-5mm}
\end{figure}

\subsection{Ablation Studies}

\subsubsection{Different number of client-side layers}
We next study the impact of different client-side layer configurations. As shown in Table~\ref{tab:local_layers}, when the number of client-side layers changes from \textit{Emb + Head} to \textit{Emb + Head + 2 Blks}\footnote{``blks'' in Tab.~\ref{tab:local_layers} means hidden layers.}, the communication cost on activations and gradients remains unchanged, while the final accuracy only fluctuates slightly. This result is consistent with the design of our framework: under the same hidden-state interface, the amount of activation and gradient transmission does not depend on how many shallow layers are placed on the client side, and therefore the communication saving remains stable across different local-layer settings. At the same time, increasing the number of local layers significantly raises client-side memory consumption. \textit{Emb + Head + 6 Blks} leads to an out-of-memory failure on our client device. This highlights an important practical constraint in edge deployment: excessively deep client-side allocation is undesirable on memory-limited devices.

\begin{table}[]
\centering
\caption{Impact of Local Layer Configuration on Client Resource Consumption (MMLU on Llama3.2-1B).}
\vspace{-2mm}
\label{tab:local_layers}
\resizebox{0.43\textwidth}{!}{%
\begin{tabular}{l|c|c|c}
\toprule
\textbf{Local Config.} & \shortstack{\textbf{Client Memory}\\\textbf{Consumption}} & \shortstack{\textbf{Comm. Cost on}\\\textbf{Act. and Grad.}} & \textbf{Accuracy} \\
\midrule
Emb + Head & 1.2 GB & 0.59 GB & 48.89\% \\
Emb + Head + 1 Blk & 2.8 GB & 0.59 GB & 48.93\% \\
Emb + Head + 2 Blks & 4.4 GB & 0.59 GB & 49.00\% \\
Emb + Head + 6 Blks & Out of Memory & - & - \\
\bottomrule
\end{tabular}%
}
\vspace{-3mm}
\end{table}

\begin{figure}[t]
\centering
\begin{minipage}[t]{0.23\textwidth}
    \centering
    \includegraphics[width=\textwidth]{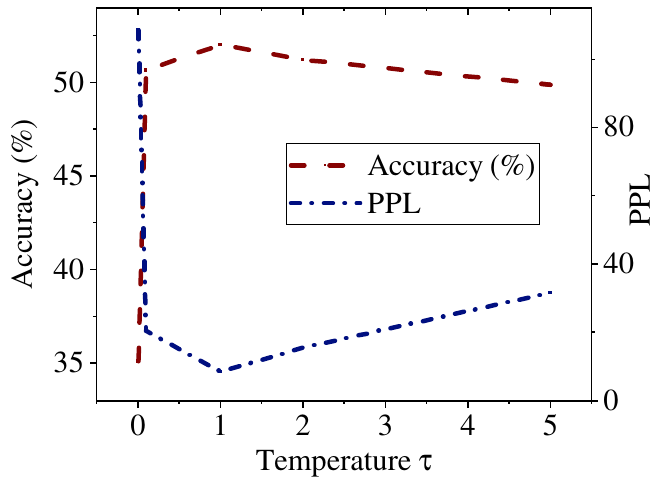}
    \vspace{-8mm}
    \captionof{figure}{Impact of $\tau$.}
    \label{fig:temperature}
\end{minipage}
\hfill
\begin{minipage}[t]{0.23\textwidth}
    \centering
    \includegraphics[width=0.9\textwidth]{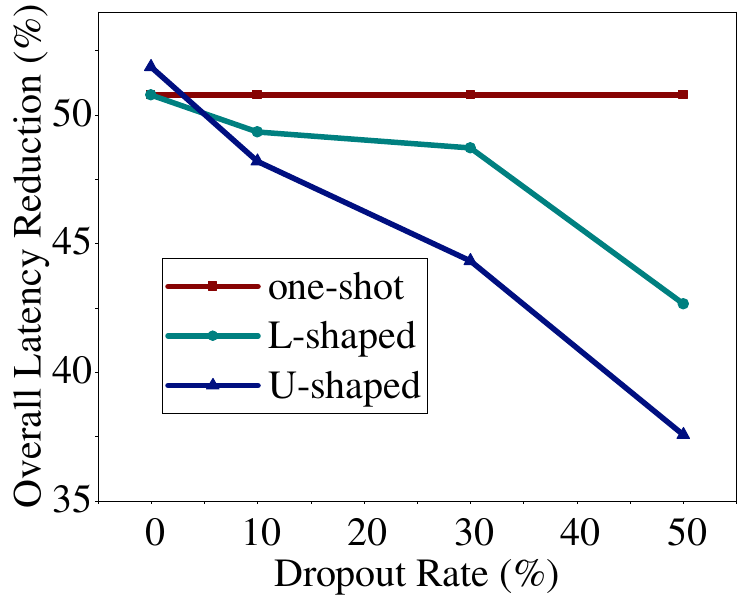}
    \vspace{-4mm}
    \captionof{figure}{Impact of dropout.}
    \label{fig:dropout}
\end{minipage}
\vspace{-5mm}
\end{figure}

\subsubsection{Impact of Temperature Coefficient $\tau$}

We next examine the impact of the temperature coefficient $\tau$ in A-Loss. As indicated by Theorem~\ref{thm:grad_equiv_rewrite}, $\tau=1$ is the theoretically most principled choice, since it recovers the same gradient form as the standard cross-entropy loss under the ideal conditions discussed in Section~4. Empirically, as shown in Fig.~\ref{fig:temperature}, the results show that model performance is best when $\tau$ is set around $1$. In contrast, when $\tau$ becomes too small, the contrastive distribution becomes excessively sharp, causing the optimization to focus too aggressively on a few dominant candidates. This makes training unstable and can even prevent convergence. On the other hand, when $\tau$ becomes too large, the distribution becomes overly smooth, weakening the contrast between the target embedding and negative embeddings. As a result, the supervision signal becomes less discriminative, and the fine-tuning effectiveness deteriorates.

\subsubsection{Robustness to Client Dropout.}

We also evaluate the robustness under client dropout, which commonly occurs in real-world edge environments due to battery depletion, network instability, or temporary device unavailability. As shown in Fig.~\ref{fig:dropout}, one-shot SFT is the most robust to dropout. This is because one-shot SFT only requires a single communication stage: as long as a client successfully comes online once and uploads all required embeddings, subsequent disconnection has little effect on training. L-shaped SFT is also relatively robust. Since the server maintains a history cache and the training process does not rely on the tight bidirectional synchronization required by U-shaped SFT, client dropout mainly reduces the amount of fresh data participation, rather than causing severe disruption to optimization. In contrast, U-shaped SFT is much more vulnerable to dropout because it depends on continuous step-level client-server interaction throughout training. Therefore, as the dropout rate increases, the performance of U-shaped SFT degrades much more sharply, while one-shot SFT remains largely unaffected and L-shaped SFT degrades more gracefully.

\vspace{-2mm}
\section{Conclusion}

In this paper, we presented L-shaped SFT, a communication-efficient split federated fine-tuning framework for LLMs. By moving supervision to the server side and removing the bidirectional bottleneck of conventional U-shaped SFT, L-shaped SFT substantially reduces communication costs and overall latency while maintaining competitive fine-tuning quality. Building on this design, we further introduced one-shot SFT, which minimizes client active time by compressing client participation into a single upload stage and allowing the remaining optimization to proceed entirely on the server. Our real-world testbed results on heterogeneous edge devices show that these designs make split LLM fine-tuning substantially more practical in bandwidth-constrained and intermittently connected environments. Although we do not address the potential privacy leakage of uploaded activations, this issue is a common challenge faced by all split fine-tuning and split federated learning frameworks. Our contribution is orthogonal to activation-level privacy protection such as differential privacy \cite{charles2024finetuning} and homomorphic encryption \cite{chen2022x}.

\vspace{-2mm}
\bibliography{citations}
\bibliographystyle{ACM-Reference-Format}

\end{document}